\documentclass[runningheads]{llncs}
\usepackage{graphicx}
\usepackage{amssymb,amsmath,amsfonts,mathrsfs}

\usepackage{xcolor}
\usepackage[linesnumbered,ruled,vlined]{algorithm2e}

\usepackage{wrapfig}

\newcommand{\tnm}[1]{{\textbf{\texttt{#1}}}} 

\SetCommentSty{mycommfont}

\SetKwInput{KwInput}{Input}                
\SetKwInput{KwOutput}{Output}              

\newcommand{\reals}{\mathbb{R}}

\newcommand{\nS}{\mathcal{S}}

\newcommand{\RS}{\texttt{RS}}

\begin{document}
\title{Reachability of Linear Uncertain Systems: Sampling Based Approaches}
%
%

\author{Bineet Ghosh\inst{1} \and
Parasara Sridhar Duggirala\inst{1}}

\institute{The University of North Carolina at Chapel Hill
\email{\{bineet,psd\}@cs.unc.edu}}

%
%

\maketitle              
\begin{abstract}
In this work, we perform safety analysis of linear dynamical systems with uncertainties.
Instead of computing a conservative overapproximation of the reachable set, our approach involves computing a \emph{statistical approximate reachable set}.
As a result, the guarantees provided by our method are probabilistic in nature. 
In this paper, we provide two different techniques to compute statistical approximate reachable set.
We have implemented our algorithms in a python based prototype and demonstrate the applicability of our approaches on various case studies. We also provide an empirical comparison between the two proposed methods and with \texttt{Flow*}.


\end{abstract}

\section{Introduction}
\label{sec:introduction}
Formal analysis of Cyber-Physical Systems (CPS) deployed in safety critical scenarios can provide rigorous safety assurances.
Such formal analysis requires a mathematically precise model of the system behavior.
While safety verification of CPS by performing reachable set computation with precise models has been widely studied~\cite{frehse2005phaver,girard2005reachability,frehse2011spaceex,bak2017hylaa}, these techniques are fragile with respect to model uncertainties.
That is, any error in the model description would invalidate all the safety guarantees obtained from reachable set computation.
Such model errors could be either because of implicit parameters in the model, sensor and measurement error, or unaccounted factors.
In this paper, we study two new techniques for performing rigorous analysis of system with modeling uncertainties such that we can make the safety analysis robust to model uncertainties.
We restrict our attention to CPS that can be modelled as linear dynamical systems.

The trajectories of linear dynamical systems can be represented in closed form using matrix exponential.
The effect of a model uncertainty on a trajectory, therefore, is a complex nonlinear function.
Given the limited scalability of nonlinear constraint solvers, for performing rigorous safety analysis, one is either required to compute an overapproximation of the reachable set either using representations or from sampled dynamics from the uncertainty. 
Such approaches have been proposed in~\cite{10.1145/1967701.1967717,DBLP:conf/emsoft/LalP15,ghosh2019robust}.
%
%
However, either the overapproximation of the reachable set is too conservative or the time taken for computing the reachable set to a required degree of precision is too expensive.

In this paper, we mitigate these two drawbacks by computing reachable sets artifacts that provide statistical guarantees.
These artifacts facilitate the user to trade-off between the computational cost and the statistical guarantees desired.
Hence, the user can compute reachable set according to a desired statistical confidence --- adjusting the computational cost --- based on the application scenario.
%
%
Additionally, the reachable sets computed using our approach are free from \textit{wrapping effect}, as the computed reachable set at a time step is independent of the previous steps' computations. 

This paper presents two techniques for computing such reachable sets.
The first technique uses a sequential hypothesis testing framework to verify that the candidate reachable set provides the desired statistical guarantees.
This framework is similar in nature to widely applied statistical model checking techniques.
Our second technique uses model learning with probably approximately correct guarantees.
In this technique, the model that approximates the reachable set of linear systems with uncertainties is learned by solving an optimization problem.
The statistical guarantees of the model are a result of the formulation of the optimization problem.
We compare the performance and accuracy of both these techniques to highlight the trade-off between computational effort and statistical guarantees using different techniques.

We have implemented our techniques in a python prototype tool and demonstrate the applicability of our approach on various standard benchmarks. We provide a empirical comparison between the two proposed approaches and also compare it with \texttt{Flow*} \cite{7809839}. Our evaluation on various benchmarks shows that artifacts that give reasonable guarantees can be computed very efficiently.

This paper is organized as follows.
In section \ref{sec:prelims} we provide the notations and the formal definitions that has been used in the rest of the paper.
In section \ref{sec:appx_props} we present our first method to compute reachable set of linear uncertain systems using hypothesis testing.
In section \ref{sec:pac} we present our second method to compute reachable set of linear uncertain systems using model learning.
In section \ref{sec:eval} we present the evaluation of our algorithm on several benchmarks, and provide a comparison between the two methods and also with \texttt{Flow*}.

\section{Related Work}
\label{sec:relatedWork}
The present work draws inspiration from previous works computing reachable sets with uncertainties, such as the following works: (i) the system is discretized, and the effect of uncertainties are computed separately \cite{10.1145/1967701.1967717}; (ii) computes flowpipes using sampling based method \cite{DBLP:conf/emsoft/LalP15}; (iii) computes exact reachable set for a subset of uncertainties \cite{ghosh2019robust}.
Unlike these works, this paper proposes two statistical approaches to efficiently compute reachable set of linear uncertain systems. And therefore provides probabilistic guarantees on the reachable set. Statistical verification has been widely used, some of such are: \cite{10.5555/647771.760735,10.1007/11513988_26,10.1007/978-3-642-03845-7_15,10.1007/978-3-642-24372-1_1}. More works on statistical verification can be found in \cite{10.1007/978-3-642-16612-9_11}.

Sampling based approaches, similar to our first approach, has been used to learn \textit{discrepancy function}\cite{6658604} in \cite{fan2017dryvrdatadriven}. Such sampling based approaches has also been used to verify hyper-properties of systems in \cite{wang2019statistical}. In a recent work \cite{10.1145/3365365.3382209}, \emph{Statistical Model Checking} based on \emph{Clopper-Pearson confidence levels} \cite{10.1145/3049797.3049804} has also been used to verify samples specifications of a Neural Network based controller, that are captured by \emph{Signal Temporal Logic (STL)} formulas. The closest work to our first approach, that uses \emph{Jeffries Bayes Factor} test, is \cite{7945001}.

Our second approach is based on learning a simpler model, that provides probabilistic guarantees, from finite number of samples. Such \emph{Probabilistic Approximately Correct (PAC)} models has been used in several works \cite{park2020pac,ashok2019pac,chen2015pac}. For our second approach, we use \emph{scenario optimization} \cite{1632303} to learn a PAC model. Such scenario optimization based techniques has also been used to find safe inputs for black box systems \cite{8882768}. The closest work to our second approach, that uses scenario optimization, is \cite{xue2020pac}.

Linear Uncertain Systems can also be modelled as a non-linear system. Some of the works that deal with computing reachable sets of non-linear systems are: \cite{7809839,10.1007/978-3-642-39799-8_18,6987596,10.1007/978-3-662-46681-0_5,10.1007/978-3-642-24690-6_13,10.1007/978-3-662-46681-0_15,10.1007/978-3-319-02444-8_37,ARCH15:An_Introduction_to_CORA}.

\section{Preliminaries}
\label{sec:prelims}
In this section we layout the definitions and the notations that are used in the rest of the paper. 
A closed interval is denoted by $[a,b]$, \emph{i.e.} $[a,b]=\{x \in \mathbb{R}~|~a \le x \le b\}$. 
The system under consideration evolves in $\reals^n$, called \emph{state space}.
For $p \in \reals$, $|p|$ denotes the absolute value of $p$.
States and vectors in $\reals^n$ are represented as $x$ and $v$ respectively, and $||x||$ denotes the Euclidean norm.
Given an $\delta > 0$ and $x\in \reals^n$, $B_{\delta}(x) = \{~y~|~ ||y - x || \leq \epsilon\}$.
For a set $S \subseteq \reals^n$, $B_{\delta}(S) = \bigcup_{x \in S} B_{\delta}(x)$.
Sometimes we also refer to $B_{\delta}(S)$ as $bloat(S, \delta)$.
%
%
%
%
Given two sets $S_1$, $S_2$, the Hausdorff distance between them $max\{ sup_{x \in S_1} inf_{y \in S_2} ||x-y||, sup_{y \in S_2} inf_{x \in S_1} ||x-y|| \}$ is denoted as $dis(S_1, S_2)$. 
Convex hull of $S_1$ and $S_2$, $\{\lambda x + (1-\lambda) y ~|~ x \in S_1, y \in S_2, 0 \leq \lambda \leq 1\}$ is denoted as $\mathsf{ConvexHull(S_1, S_2)}$

Given a matrix $M \in \mathbb{R}^{m \times n}$, the $(i,j)^{th}$ element is denoted as $M[i,j]$. 
%
%
We overload the operator $|.|$ for matrices as well. Given a matrix $M \in \mathbb{R}^{n \times m}$, $|M|$ is a matrix such that, $\forall_{i,j}$ $|M|[i,j] = |M[i,j]|$.

\begin{definition} [Matrix Norms]
\label{def:matNorms}
Given a matrix $M \in \mathbb{R}^{n \times n}$, its matrix 2 norm is denoted as $||M||_2$ \emph{i.e.} $||M||_2=\sigma_{max}(M)$, and $||M||_F$ is the matrix Frobenius norm. We also use $||M||$ to denote $||M||_p$, where $p$ can be anything in $\{2, F\}$. For any matrix $M$, $||M||_2 \le ||M||_F$. 
\end{definition}
The maximum singular value of a matrix $M$ is denoted by $\sigma_{max}(M)$.
The domain of Boolean matrices of dimension $m \times n$ are denoted with $\mathbb{B}^{m \times n}$. 
The addition and multiplication operations on Boolean matrices are extensions of regular addition and multiplication operations with addition being the disjunction operation and multiplication being the conjunction operation.

\begin{definition}[Continuous Linear Dynamical Systems]
\label{def:contSys}
Given a matrix $A \in \mathbb{R}^{n\times n}$, a continuous linear dynamical system is denoted as $\dot{x} = Ax$
\end{definition}

\begin{definition}[Trajectories]
\label{def:traj}
A trajectory of the continuous linear dynamical system, denoted as $\xi_{A} : \mathbb{R}^n \times \mathbb{N} \rightarrow \mathbb{R}^n$, describes the evolution of the system in time. Given an initial state $x_0 \in \mathbb{R}^n$, the trajectory is defined as $\xi_{A}(x_0, t) = e^{At} x_0$
\end{definition}
We drop $A$ from the subscript of $\xi$ when it is clear from the context.

\begin{definition}[Reachable Set]
Given a linear dynamical system $\dot{x} = Ax$, initial set of states $\theta$, and a time step $t \in \mathbb{N}$, the reachable set of states 
\begin{equation}
\texttt{RS}(\theta, A, t) = \{~ x~|~\exists x_0 \in \theta, x = \xi_{A}(x_0, t)\}  
\end{equation}
\end{definition}
In this paper, we represent the reachable set of a linear system as a star set.

\begin{definition}[Star] \cite{10.1007/978-3-319-41528-4_26}
\label{def:stars}
A generalized star $S$ is defined as a tuple $\langle c, V, P \rangle$ where $c \in \mathbb{R}^n$ is called the {\em anchor}, $V = \{v_1, v_2, \ldots, v_m\}$ where $\forall i, 1 \leq i \leq m, v_i \in \mathbb{R}^n$ are called a set of {\em generators} (that span $\mathbb{R}^n$), and $P: \mathbb{R}^m \rightarrow \{\top, \bot\}$ is called the {\em predicate}. The set of states represented by a generalized star is defined as:
\begin{eqnarray}
&&[\![S]\!] = \{~x~|~ \exists \alpha_1, \alpha_2, \ldots, \alpha_m \mbox{ such that }~~~~~~~~~~~~~~~~~~  \nonumber \\
&&~~~~~~~~~~~~~~~ x = c + \Sigma_{i=1}^{m}\alpha_i v_i \mbox{ and } P(\alpha_1, \alpha_2, \ldots, \alpha_m) = \top  \}
\end{eqnarray}
\end{definition}
We abuse notation and use $S$ to refer to both $[\![S]\!]$ and $S$.

\begin{definition}[from~\cite{10.1007/978-3-319-41528-4_26}]
\label{def:reachSetStar}
Given a linear dynamical system $\dot{x} = Ax$ and initial set $\theta$ represented as a star set $\theta \triangleq \langle c, V, P \rangle$, the reachable set $\texttt{RS}(\theta, A, t)$ is also a star with anchor $e^{At} c$, generators $V' = \{e^{At}v_1, e^{At}v_2, \ldots, e^{At}v_m\}$, and the same predicate $P$ as $\theta$.
$$
\RS(\theta, A, t) \triangleq \langle e^{At} c, V', P\rangle 
$$
\end{definition}

\subsection{Linear Uncertain Systems}

\begin{definition}
\label{def:pme}
Given a set of symbolic variables $Vars = \{y_1, \ldots, y_k\}$, we denote all possible polynomial expressions over symbolic variables as \emph{PE}. 
A matrix $M$ is called a \emph{polynomial matrix expression} if $M \in \{\reals \cup \mathsf{PE}\}^{n \times n}$. 
That is, the elements of the matrix $M$ are not just real numbers, but can also be polynomial expressions.
\end{definition}
A \emph{polynomial matrix expression} $M$ can also be represented as a polynomial over $Vars$, where, the coefficients of each monomial is a matrix instead of a real value.
Polynomial matrix expressions which do not have any second or higher degree terms are called \emph{linear matrix expressions}.
\begin{definition}[Uncertain Linear Systems and Reachable Set] 
An uncertain linear dynamical system, denoted as $\dot{x} = \Lambda x$ where $\Lambda \subseteq \reals^{n \times n}$.
We represent $\Lambda$ as $\Lambda = \langle M_{\Lambda}, D_{\Lambda} \rangle$ where $M_{\Lambda}$ is a polynomial matrix expression over $Vars$ and $D_{\Lambda} \subseteq \reals^k$. The set of matrices represented by $\Lambda$ are
$$
\Lambda = \{A~|~ A = M_{\Lambda} [y_1 \leftarrow \gamma_1, \ldots, y_k \leftarrow \gamma_k] \wedge (\gamma_1, \ldots, \gamma_k) \in D_{\Lambda} \}.
$$
Where $M_{\Lambda} [y_1 \leftarrow \gamma_1, \ldots, y_k \leftarrow \gamma_k]$ denotes the evaluation of the matrix polynomial expression $M_{\Lambda}$ with $y_1$ assigned the value $\gamma_1$, $y_2$ assigned $\gamma_2$, \ldots, and $y_k$ assigned the value $\gamma_k$, respectively.
\end{definition}
In this paper, we only consider uncertain systems where $D_{\Lambda}$ is a bounded polytope.

%

\begin{definition}
\label{def:sampledynamics}
Given an uncertain linear system $\dot{x} = \Lambda x$, we say that $\dot{x} = Ax$ as a sample dynamics of the uncertain linear system if $A \in \Lambda$.
We denote the dynamics obtained by assigning variables $y_1, y_2, \ldots, y_k$ the values $\gamma_1, \gamma_2, \ldots, \gamma_k$ respectively as $\dot{x} = \Lambda_{\overline{\gamma}} x$.
\end{definition}

\begin{example}
Consider an uncertain linear system $\dot{x} = \Lambda x$ where $\Lambda = \langle M_{\Lambda}, D_{\Lambda} \rangle$ where 
$M_{\Lambda} = \begin{bmatrix}
    ac  & 2a \\
    0   & b^2 
\end{bmatrix}
$ is a polynomial matrix expression over $Vars = \{a, b, c\}$ $D_{\Lambda} = \{a \in [0,1] \wedge b \in [0,1] \wedge c \in [0,1]\}$.
One can represent the polynomial matrix expression as $M_{\Lambda} = M_0 + M_1 a + M_2 ac + M_3 b^2$ where $M_0$ is the zero matrix, 
$M_1 = \begin{bmatrix}
    0  & 2 \\
    0   & 0 
\end{bmatrix}
$,
$M_2 = \begin{bmatrix}
    1  & 0 \\
    0   & 0 
\end{bmatrix}
$, and 
$M_3 = \begin{bmatrix}
    0  & 0 \\
    0   & 1 
\end{bmatrix}
$.
A sample dynamics obtained by assigning $a=1, b = 1,$ and $c=1$, i.e., $\dot{x} = \Lambda_{(1,1,1)} x$ is 
$\dot{x} = \begin{bmatrix}
    1  & 2 \\
    0   & 1 
\end{bmatrix} x$
\end{example}


%
%
One of the routinely encountered form of uncertain systems is where $M_{\Lambda}$ is a linear matrix expression and $D_{\Lambda}$ is a hyper-rectangle in $\reals^k$. In such cases, the uncertain linear system can be represented as an interval matrix.


%


\begin{definition}[Interval Matrices]
\label{def:intervalMatrices}
An interval matrix $\Omega \subset \mathbb{R}^{m \times n}$, where $\forall_{i,j} \Omega[i,j]=[a_{i,j},b_{i,j}]$ and $[a_{i,j},b_{i,j}]$ is a closed interval. 
\end{definition}



\begin{definition}[Matrix Support (from~\cite{ghosh2019robust})] 
\label{def:matrixStructure}
Given a matrix $M \in \mathbb{R}^{m \times n}$, $\texttt{supp}(M) = B$ where $B \in \mathbb{B}^{m \times n}$ such that for all $i, 1 \leq i \leq m, 1 \leq j \leq n$, $B[i,j] = 0$ if and only if $M[i,j] = 0$.
\end{definition}

\begin{definition}[Sub-Support and Super-Support]
\label{def:subSupport}
Given Boolean matrices $B_1, B_2 \in \mathbb{B}^{m\times n}$, we say that $B_1$ is a \emph{sub-support} of $B_2$, denoted as $B_1 \leq B_2$ if and only if for all $i,j$, 
if $B_1[i,j] = 1$ then $B_2[i,j] = 1$. An equivalent formulation is for all $i,j$, 
if $B_2[i,j] = 0$ then $B_1[i,j] = 0$. We also say that $B_2$ is a \emph{super-support} of $B_1$.
\end{definition}

\section{Reachable Sets With Probabilistic Guarantees From Hypothesis Testing}
\label{sec:appx_props}

%
In this section, we present a technique where the probabilistic guarantees provided by the reachable set is verified by using hypothesis testing.
This approach involves three sub-routines named: \textsf{Generator}, \textsf{Verifier}, \textsf{Refiner}.
The \textsf{Generator} sub-routine computes a candidate reachable set by performing operations on reachable set of a dynamics sampled from $\Lambda$.
The \textsf{Verifier} sub-routine uses statistical verification techniques, more specifically, \emph{Jeffries Bayes Factor} test, to validate that the candidate set provided by \textsf{Generator} contains the reachable set of dynamics from $\Lambda$ with high probability.
The \textsf{Refiner} sub-routine is only invoked if the \textsf{Verifier} infers that the candidate set fails the test.
In such cases, the \textsf{Refiner} reuses the artifacts generation during the verification phase and invokes \textsf{Generator} with additional arguments so that it will generate a different candidate set.
The workflow is given in Figure~\ref{fig:meth}.


\begin{wrapfigure}{r}{0.5\textwidth}
  \begin{center}
    \includegraphics[width=5.5cm,height=3.9cm]{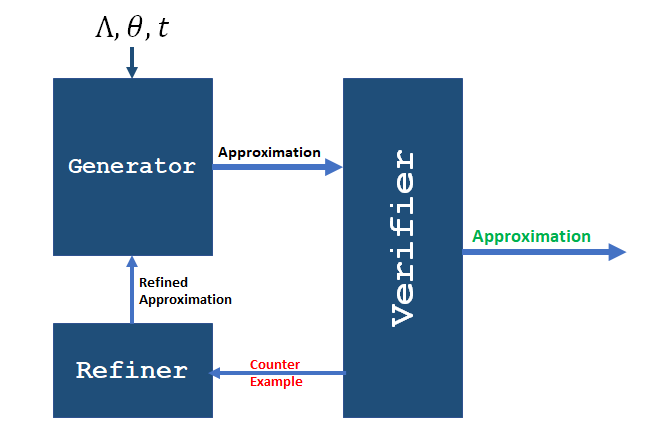}
  \end{center}
  \caption{Workflow for computing probably approximately correct reachable set for uncertain linear systems.}
\label{fig:meth}
\end{wrapfigure}

%
%
Suppose that for the uncertain linear system $\dot{x} = \Lambda x$, the valuations to variables $Vars$ is drawn according to a probability distribution $\mu$.
%
%
Additionally, suppose that $\forall \overline{\gamma} \notin D_{\Lambda}, \mu(\overline{\gamma}) = 0$ and $\int_{\overline{\gamma} \in D_{\Lambda}} \mu(\overline{\gamma}) d\overline{\gamma} = 1$.
That is, the probability density function is zero outside the domain $D_{\Lambda}$.

\begin{definition}
\label{def:probReach}
Given a linear uncertain system $\dot{x} = \Lambda x$, an initial set $\theta$, and time $t$, a set $\nS \subseteq \reals^n$ is called a $p$-reachable set at time $t$ where: 
$$
\int_{RS(\theta, \Lambda_{\overline{\gamma}},t) \subseteq \nS} \mu(\overline{\gamma}) d\overline{\gamma} = p.
$$

Given a set $\nS$, we denote this probability as $Prob(\Lambda, \theta, t, \nS)$.
\end{definition}

Observe that $\reals^n$ is a 1-reachable set for any initial set for any uncertain linear system and the emptyset is a 0-reachable set for any initial set for any uncertain system.
Given a $\nS$, and time $t$, it is challenging to compute $Prob(\Lambda, \theta, t, \nS)$, the fraction of reachable sets from $\Lambda$ that are contained in $\nS$.
Therefore, instead of computing the probability of the reachable set being contained in $\nS$, we determine whether the probability is above a certain threshold with high confidence.
More specifically, given a threshold $c$, we estimate that the probability of Type 1 error i.e., inferring that $Prob(\Lambda, \theta, t, \nS) \geq c$ whereas in reality $Prob(\Lambda, \theta, t, \nS) < c$ is bounded by $\delta$.

\begin{definition}
\label{def:pac-reachability}
Given an uncertain system $\dot{x} = \Lambda x$, initial set $\theta$, time $t$, threshold $c \in [0,1]$, and probability bound $\delta$, the set $\nS$ is a $c,\delta$  probably approximate reachable set if
$
Type1Err[ Prob(\Lambda, \theta, t, \nS) \geq c] < \delta.
$
\end{definition}
The rest of this section is organized as follows. In Subsection \ref{sec:appx_comp}, we discuss several heuristics to generate the candidate reachable set (the \tnm{Generator} module). 
In Subsection \ref{sec:stat_ver}, we present the \emph{Statistical Hypothesis Testing} framework that is used for validating the properties of the candidate reachable set (the \tnm{Verifier} module).
The full description of the algorithm that invokes \tnm{Generator}, \tnm{Verifier} and \tnm{Refiner} is given in Subsection \ref{sec:main_algo}.

\subsection{Generator: Computing Candidate Reachable Set}
\label{sec:appx_comp}
In this subsection, we present several heuristics for computing the candidate for provably approximate reachable set.
Since the validity of a candidate reachable set would be evaluated using statistical hypothesis testing framework, one can use any heuristic for generating a candidate reachable set.
A typical candidate would be generated by computing reachable set of dynamics sampled from $\Lambda$ in various ways.
%
%
%
We present five informed heuristics for computing the candidate reachable set.
One of the heuristics, more specifically, the heuristic that generates the candidate based on \emph{structure of uncertainties} is a \textbf{conservative overapproximation} of the reachable set of uncertain system.
For each of the heuristics we propose, we provide a basic rationale for proposing such a heuristic and also provide theoretical guarantees if any of such a heuristic.
%

\subsubsection{Bloating the Reachable Set of Mean Dynamics}
\label{subsec:centerBased}


%

One way to model the uncertainties in the dynamics is to construct a \emph{mean} dynamical system and consider the uncertainties as external inputs that change as a function of the state.
In such instances, one can compute the reachable set by \emph{bloating} the reachable set of the mean dynamical system by the appropriate value.
For computing rigorous enclosure of reachable set, one can use the exponential of the Lipschitz constant of the dynamics to compute this bloating factor.
However, most often, such analysis results in an overapproximation that is too conservative and is often not useful for proving safety.

In this heuristic, we first compute the reachable set of the mean dynamical system.
We then compute the Hausdorff distance of the mean reachable set and the reachable set of a random sample of $N$ dynamics from $\dot{x} = \Lambda x$.
%
We then bloat the reachable set of the mean dynamics by the computed Hausdorff distance and an additional $\epsilon$ that is pre-determined by the user.
The precise description of generating the candidate set is given below.

\begin{enumerate}
    \item We compute the \emph{mean dynamics} $A_c \in \Lambda$. If the uncertainties in the dynamics are intervals, then, this mean dynamics is often the matrix obtained by substituting all the intervals with the corresponding mid-points. 

 \item Let, $S=\{A_1, A_2, \cdots, A_N\}$ be a set of $N$ random \emph{sample dynamics} of $\Lambda$, \emph{i.e.} $\underset{1 \le i \le N}{\forall} A_i \in \Lambda$. 
%

 \item Let, 
$
d = \underset{A \in S}{\texttt{max}}~\Big\{ \texttt{dis}\big(\texttt{RS}_{\theta}(A_{c},t),\texttt{RS}_{\theta}(A,t) \big) \Big\}
$. 
Where, $\texttt{dis}(S_1,S_2)$ is defined as the \emph{Hausdorff distance} between $S_1$ and $S_2$. 

 \item Compute $\texttt{Bloat}\big(\texttt{RS}_{\theta}(A_{c},t),d+\epsilon\big)$, where $\epsilon > 0$ is specified by the user.
 
\end{enumerate}


We provide an example illustrating the above technique in Example \ref{ex:reachmean} (Appendix). The reachable set returned by this procedure is denoted as \\ $\mathsf{bloatMean}(\Lambda, \theta, t)$.

\subsubsection{Structure Guided Reachable Set Computation}
\label{subsec:sepBased}



We now present another heuristic for generating candidate reachable set by sampling the uncertain system based on the structure of the coefficients in the uncertain matrix polynomial expression.
In~\cite{ghosh2019robust}, the authors only considered uncertain matrices as linear matrix expressions and present sufficient conditions where the linear matrix expressions are closed under multiplication.
We now present a sampling based reachable set computation technique that computes a provable overapproximation of the uncertain systems considered in~\cite{ghosh2019robust}.
However, if the conditions provided in~\cite{ghosh2019robust} are not satisfied, then our reachable set need not be a conservative overapproximation.
Hence, for such instances, one has to perform statistical verification, similar to the verification of the candidate reachable sets for other heuristics.

We follow~\cite{ghosh2019robust} and separate the variables in uncertain linear system as \emph{LME preserving} and \emph{non LME preserving}.
For the LME preserving variables, we sample all the vertices of the domain $D_{\Lambda}$ and grid based uniform sampling for all non LME preserving variables.
Similar to the reasons presented in \emph{Mean Dynamics} based heuristic, we compute an oriented rectangular hull overapproximation of the reachable sets of the sample dynamics to generate the candidate reachable set.

\begin{definition}
\label{def:lmepreserving}
Given an uncertain linear system $\dot{x} = \Lambda x$, a variable $y_i \in Vars$ is called \emph{LME preserving} for $M_{\Lambda}$ if and only if the following conditions are satisfied.
\begin{enumerate}
\item The variable $y_i$ only has terms of degree one in $M_{\Lambda}$.
\item $supp(M_0) \times supp(M_0) \leq supp(M_0)$, where $M_0$ is the constant term in the matrix polynomial expression $M_{\Lambda}$.
\item $supp(M_0) \times supp(M_i) + supp(M_i) \times supp(M_0) \leq supp(M_i)$ where $M_i$ is the coefficient matrix of the variable $y_i$ in $M_{\Lambda}$. 
\end{enumerate}
All the variables that do not satisfy these conditions are considered as \emph{non LME preserving}.
\end{definition}


We compute the oriented rectangular hull overapproximation of the reachable sets of dynamics sampled from the vertices of $D_{\Lambda}$ and for non LME preserving variables, we generate additional sample dynamics, generated at regular intervals in $D_{\Lambda}$.
The steps for computing the candidate reachable set using this heuristic are given below.


\begin{enumerate}
\item Compute all LME preserving variables in $Vars$.

\item Let $Vertices_{\Lambda} = \{A_1, A_2, \ldots, A_M\}$ be the sample dynamics obtained by evaluating the uncertain linear system on the vertices of $D_{\Lambda}$. 
Additionally, for all non LME preserving variables, generate $K$ samples at regular intervals for each variable from $D_{\Lambda}$, called $Uniform_{\Lambda}$.

\item Compute the reachable set of all the dynamics in $Vertices_{\Lambda}$ and $Uniform_{\Lambda}$ at time $t$. 


\item Compute the oriented rectangular hull overapproximation of all the reachable sets computed in $Vertices_{\Lambda}$ and $Uniform_{\Lambda}$. 
\end{enumerate}

Given a linear uncertain dynamics $\dot{x}=\Lambda x$, an initial set $\theta$ and time $t$, we denote a reachable set computed using this heuristic as \texttt{uniformORH}($\Lambda, \theta, t$).

%
%

If all the variables in $\Lambda$ are LME preserving, then $\texttt{uniformORH}(\Lambda, \theta, t)$ is an overapproximation of the reachable set of all the dynamics in $\Lambda$ for initial set $\theta$. We provide a proof of this is Section \ref{appxsubsubsec:conservative}

\subsubsection{Sampling Based on Singular Values}
\label{subsec:maxSVBased}

Building on the heuristic using \emph{Mean Dynamics}, it is unclear why one should pick the reachable set of mean dynamics and bloat it.
We enquire whether there is a more \emph{appropriate} sample reachable set that can be bloated to compute the candidate reachable set.

As mentioned earlier, one can model uncertainties as external inputs to a dynamical system.
Consider the dynamical system $\dot{x} = Ax + U$ where $U$ is an input from the environment.
If the external input $U$ stays constant and if $A$ is invertible, the trajectory for the system with input can be written as $\xi(x_0,t) = e^{At} x_0 + A^{-1}(e^{At} - I) U$.
The effect of the input is therefore magnified by a factor of $e^{At}$.
Therefore, a more appropriate candidate for bloating of reachable set would be one where the magnitude of $e^{At}$ is higher.
Since we know that the singular values of a matrix directly effects the magnitude of matrix exponential, we pick the matrix with higher singular value.
Once such a candidate from $\Lambda$ has been selected, we bloat it in a way similar to \emph{Mean Dynamics} based heuristic.

Here, we restrict our attention to uncertain systems represented as interval matrices and leverage prior work from the literature~\cite{farhadsefat2011norms}.
%
%
%
%
In \cite{farhadsefat2011norms}, the notation of matrix measure is extended to interval matrices, such as $\Lambda$, when represented as an interval matrix. 
Given an interval matrix $\Lambda$, the \emph{norm-2} of $\Lambda$, $||\Lambda||$ is defined as
$
||\Lambda||_2 = \texttt{sup} \big\{~||A||_2~|~A~\in \Lambda \big\}
$.

Let, $A_{max} \in \Lambda$ be an element in $\Lambda$ with the maximum singular value, \emph{i.e.} $\forall_{A \in \Lambda} \sigma_{max}(A_{max}) \ge \sigma_{max}(A)$; i.e., $\sigma_{max}(A_{max}) = || \Lambda ||_2$

Theorem 7 in \cite{farhadsefat2011norms} provides an algorithm to find $|| \Lambda ||_2$ as follows:
\begin{equation}
\label{eq:maxSVThm}
    ||\Lambda||_2
    =
    \underset{|y|=e_n, |z|=e_n}{\texttt{max}}
    ||
    A_c + (yz^T) \circ \Delta
    ||_2
\end{equation}
Where $A_c$ is obtained by substituting each element in $\Lambda$ with the mid-point of the interval and $\Delta$ is obtained by substituting each element in $\Lambda$ with \emph{half the width} of the interval.
%
%
We apply Equation \ref{eq:maxSVThm}, and compute $A_{max}$ by considering all $2^{n}$ possibilities for $y$ and $2^n$ possibilities for $z$, respectively.
After computing the matrix with the largest singular value, we bloat this in a manner similar to the bloating of reachable set of mean dynamics. 
The steps for computing the candidate reachable set using the maximum singular value matrix are as follows:
%

\begin{enumerate}
\item Compute the \emph{maximum singular value matrix} $A_{max} \in \Lambda$ using the technique described in~\cite{farhadsefat2011norms} and compute the reachable set $\texttt{RS}_{\theta}(A_{max},t)$.
%
\item Let, $S=\{A_1, A_2, \cdots, A_N\}$ be a set of $N$ random \emph{sample dynamics} of $\Lambda$, \emph{i.e.} $\underset{1 \le i \le N}{\forall} A_i \in \Lambda$. 
Let, 
$
d = \underset{A \in S}{\texttt{max}}~\Big\{ \texttt{dis}\big(\texttt{RS}_{\theta}(A_{max},t),\texttt{RS}_{\theta}(A,t) \big) \Big\}
$. Where, $\texttt{dis}(S_1,S_2)$ is the \emph{Hausdorff distance} between the two sets. 

\item Compute $\texttt{Bloat}\big(\texttt{RS}_{\theta}(A_{max},t),d+\epsilon\big)$, where $\epsilon > 0$ is provided by the user.
\end{enumerate}

Given a linear uncertain dynamics $\dot{x}=\Lambda x$, an initial set $\theta$ and time $t$, we denote the candidate reachable set computed using this method as \texttt{maxSV}($\Lambda, \theta, t$). We provide an example illustrating the above technique in Example \ref{ex:reachSV} (Appendix)

\subsubsection{Convex Overapproximation Using Oriented Rectangular Hulls}
\label{subsec:orhBased}


Instead of bloating the reachable set from a specific sample from the uncertain linear system, it is natural to extend it to a random set of samples generated from the uncertain system. 
After computing the reachable sets from samples generated, we use \emph{Oriented Rectangular Hulls} to aggregate the reachable sets. 
Oriented rectangular hulls (proposed in~\cite{10.1007/3-540-36580-X_35}) is a parallelotope where the bounding half-planes are chosen after performing principle component analysis (PCA) of various sample trajectories.
After performing PCA of the sample trajectories to compute the template directions, we compute a convex overapproximation of the union of sample reachable sets by using linear programming.



%

\begin{enumerate}
\item Let, $S=\{A_1, A_2, \cdots, A_N\}$ be a set of $N$ random \emph{sample dynamics} of $\Lambda$, \emph{i.e.} $\underset{1 \le j \le N}{\forall} A_j \in \Lambda$. A random \emph{sample dynamics} $A_j \in \Lambda$ is obtained by assigning a random set of value (within the given range of the variable) to the uncertain variables of $\Lambda$.

\item Compute reachable set of all the \emph{sample dynamics} in $S$:
$\mathscr{R}$=\{$\texttt{RS}_{\theta}(A_1,t)$, $\texttt{RS}_{\theta}(A_2,t)$, $\cdots$, $\texttt{RS}_{\theta}(A_N,t)$\}.

\item Generate $h$ randomly sampled points from each reachable set $\RS_{\theta}(A_i, t)$ and compute the \emph{Oriented Rectangular Hull}~\cite{10.1007/3-540-36580-X_35} that selects the template directions by performing the PCA of the sample points and then compute the upper and lower bounds for each template direction.

\item For each template direction $u_{i}^{T}$ in the oriented rectangular hull, compute $max_{i,j}$ and $min_{i,j}$ as 
\begin{eqnarray*}
max_{i,j} &=& \mathsf{max}~ u_{i}^{T}x \mbox{ for } x \in \RS_{\theta}(A_j,t). \\
min_{i,j} &=& \mathsf{min}~ u_{i}^{T}x \mbox{ for } x \in \RS_{\theta}(A_j,t).
\end{eqnarray*}
using linear programming.

\item For each template direction $u_{i}^T$, compute $ub_i = \mathsf{max}_{j \in 1 .. N}\{ max_{i,j} \}$ and $lb_i = \mathsf{min}_{j \in 1 .. N}\{ min_{i,j} \}$.

\item Represent the oriented rectangular hull overapproximation as $ \wedge_{i = 1}^n lb_i - \epsilon \leq u_i^T x \leq ub_i + \epsilon$, where $\epsilon > 0$ is provided by the user. 

\end{enumerate}

Given an uncertain linear system $\dot{x} = \Lambda x$, an initial set $\theta$ and time $t$, we denote the candidate reachable set using this heuristic as $\texttt{encloseORH}(\Lambda, \theta, t)$.


One specific instance of computing the candidate reachable set would be to use the orthonormal directions for the template directions, instead of performing PCA on sample trajectories.
For such template directions, one need not invoke optimization problem as performed in step 4), instead one can compute an axis aligned overapproximation using interval arithmetic.
Additionally, performing convex hull of such axis aligned sets is computationally very cheap.
We represent such axis aligned reachable sets as $\texttt{boxBloat}(\Lambda, \theta, t)$.

Figure \ref{fig:exORH} (Appendix) shows an example of \emph{Oriented Rectangular Hull}  computed around a set of reachable sets for one of the benchmarks considered in this paper. We also provide an example illustrating the above technique in Example \ref{ex:reachORH} (Appendix).

\subsection{Statistical Verification of Candidate Reachable Sets}
\label{sec:stat_ver}
In this subsection, we present the statistical hypothesis testing framework that we use for proving that the candidate reachable set is an overapproximation of the reachable set of uncertain linear system with high probability.
%
%
More specifically, given a candidate reachable set $\RS$, confidence threshold $c$, and error $\delta$, we employ \emph{Jeffries Bayes Factor} test \cite{7945001} to verify if
$
P[ Prob(\Lambda, \theta, t, \RS) < c] < \delta.
$
%
We employ Bayes Factor test because of two reasons.
First, this check requires computing the reachable set for sample dynamics generated from $\Lambda$, which is efficient.
Second, the number of samples to be verified using this technique is independent of the number of uncertainties in $\Lambda$ or the size of the domain $D_{\Lambda}$.  
This is in contrast to other reachable set computation techniques where the number of uncertainties and the size of the domain $D_{\Lambda}$ affect the computation time for reachable set.
%
%
The algorithm we propose in this method will be the \tnm{Verifier} subroutine in the overall workflow.

In \emph{Statistical Hypothesis Testing}, one is required to form two hypotheses --- the null hypothesis $H_0$ and the alternate hypothesis $H_1$. 
After the hypotheses are formalized, evidence is gathered by randomly sampling (finite number of samples) the sample space --- once sufficient evidence has been gathered to support either of the hypotheses \emph{i.e.} $H_0$ or $H_1$, the algorithm terminates by \emph{accepting} the hypothesis, which has been supported by the observed random samples and \emph{rejecting} the other one. The criterion for \emph{accepting/rejecting} a hypothesis is based on the method one chooses. 
In this work, we choose a \emph{Bayesian approach} --- \emph{Jeffries Bayes Factor} test as in \cite{7945001}, as it is suitable and yet simple for our purpose. 
The probability of wrongly \emph{accepting} $H_1$, when $H_0$ is the actual truth, is known as the \emph{type I error}.
Informally, if $H_1$ states that the given \emph{approximation property} is true ($H_0$ states otherwise) ---  for our answer to be credible, the value of \emph{type I error} has to be sufficiently low.

Given a confidence value $c \in [0,1]$, our \emph{null hypothesis} $H_0$ and \emph{alternate hypothesis} $H_1$ are as follows:
\begin{eqnarray}
H_0: Prob(\Lambda, \theta, t, \RS) < c \\
H_1: Prob(\Lambda, \theta, t, \RS) \geq c
\end{eqnarray}
Using Bayes Factor Test, we can statistically verify that $H_0$ holds with very low probability. That is, the proposed reachable set is an overapproximation of the reachable set with high confidence.

Let, $A_1, A_2, \cdots, A_K$ be a set of $K$ \emph{sample dynamics} of $\Lambda$, \emph{s.t.} 
$\underset{1 \le j \le K}{\forall} A_j \in \Lambda \land 
\RS_{\theta}(A_i, t) \subseteq \RS
$.
The probability of this happening, given the \emph{null hypothesis} $H_0$ is true, is given by:
$$
\mathbb{P}\big[
\bigwedge\limits_{j=1}^{K} \RS_{\theta}(A_i, t) \subseteq \RS ~
|~H_0
\big]
=
\int_{0}^c q^K ~ dq
$$
Similarly, the probability of $A_1, A_2, \cdots, A_K$ satisfies $\RS_{\theta}(A_i, t) \subseteq \RS$, given $H_1$ is true is:
$$
\mathbb{P}\big[
\bigwedge\limits_{j=1}^{K} \RS_{\theta}(A_i, t) \subseteq \RS ~
|~H_1
\big]
=
\int_{c}^1 q^K ~ dq
$$
Given the above two probabilities, the \emph{Bayes Factor} is given by the ratio of the above two probabilities:
\begin{equation}
\label{eq:BayesFact}
    \frac{\mathbb{P}\big[
\bigwedge\limits_{j=1}^{K} \RS_{\theta}(A_i, t) \subseteq \RS ~
|~H_1
\big]}{\mathbb{P}\big[
\bigwedge\limits_{j=1}^{K} \RS_{\theta}(A_i, t) \subseteq \RS ~
|~H_0
\big]}=\frac{1-c^{K+1}}{c^{K+1}}
\end{equation}


Intuitively, the \emph{Bayes Factor} is the strength of evidence favoring one hypothesis over another. A sufficiently high value of $B$ indicates that the evidence favors $H_1$ over $H_0$. Given $B$, we can compute $K$   from Equation \ref{eq:BayesFact} (as given in \cite{7945001})
\begin{eqnarray}
\frac{1-c^{K+1}}{c^{K+1}} &>& B \\
K &>& -\frac{log(B+1)}{log(c)} \label{eq:K}
\end{eqnarray}

Given $B$ and $c$, \cite{7945001} gives the formula for calculating \emph{type I error} (falsely accepting $H_1$, when $H_0$ is the actual truth), denoted as $err(B,c)$ as follows:
\begin{equation}
\label{eq:typeIerror}
    err(B,c)=\frac{c}{c+(1-c)B}
\end{equation}
Note that, $err(B,c)$ is therefore the probability that our answer is wrong --- we infer $Prob(\Lambda, \theta, t, \RS) \geq c$, when in reality $Prob(\Lambda, \theta, t, \RS) < c$.

Our algorithm, as used by the \tnm{Verifier} module to \emph{accept/reject} $H_0$ or $H_1$ is therefore given as follows:

\begin{enumerate}
 \item For a chosen (sufficiently) high value of $B$ and $c$, compute $K$ as in Equation \ref{eq:K}.

 \item Let, $S=\{A_1, A_2, \cdots, A_K\}$ be a set of $K$ random \emph{sample dynamics} of $\Lambda$ generated according to the probability density function $\mu$.

 \item If all the $K$ samples in $S$ satisfies:
$
\RS_{\theta}(A_i, t) \subseteq \RS
$
 --- we accept $H_1$.

\item Otherwise, we \emph{reject} $H_1$ and accept $H_0$ with a \emph{counter example} $\mathscr{C}$. Where $\mathscr{C}=\texttt{RS}_{\theta}(A',t)$, such that $\RS_{\theta}(A',t) \nsubseteq \RS$.
\end{enumerate}
In the above mentioned algorithm (steps 1-4), the probability of falsely accepting $H_1$, when $H_0$ is the real truth (\emph{type I error}) is given by $err(B,c)$.




\subsection{Main Algorithm}
\label{sec:main_algo}
In this subsection, we combine all the sub-routines namely \tnm{Generator} (Subsection \ref{sec:appx_comp}), \tnm{Verifier} (Subsection \ref{sec:stat_ver}) and \tnm{Refiner}. 
The input to our algorithm is: (i) an uncertain dynamics $\Lambda$, (ii) an initial set $\theta$, (iii) time $t$. Before we go on to describe our algorithm, we briefly recall the functions of three sub-routines:

\begin{enumerate}
\item $\tnm{Generator}$: When given as input the uncertain dynamics $\Lambda$, initial set $\theta$, and time $t$, the generator provides a candidate reachable set that uses one of the four heuristics as described in 
%
Subsection \ref{sec:appx_comp}.
\item \tnm{Verifier}: When given the input the property  of interest (whether the reachable set is contained in the candidate reachable set or is from bounded distance from it), the confidence value (generally $c=0.99$), and the threshold of type I error, the \tnm{Verifier} either accepts the candidate reachable set or rejects it by producing an instance of the dynamics $A'$ that does not satisfy the property of interest (that is, the reachable set is not contained in the candidate set or not within the bounded distance from candidate set). 
\item \tnm{Refiner}: Given the candidate reachable set $\RS$ rejected by the verifier and the instance of reachable set that violates the property, the refiner increases the value of $\epsilon$ used in the bloating of each set such that $\epsilon > \texttt{dis}(\RS, \RS_{\theta}(A',t))$. This new value of $\epsilon$ is given as input to the generator in recomputing the candidate reachable set.
\end{enumerate}

Given the three modules \tnm{Generator}, \tnm{Verifier} and \tnm{Refiner} --- our main algorithm, an ensemble of the three modules, is given in Algorithm \ref{algo:main_algo}.

\begin{algorithm}
\caption{Reachable set computation}
\label{algo:main_algo}
\KwInput{Uncertain dynamics $\Lambda$, an initial set $\theta$, time $t$, $\epsilon > 0$.}
\KwOutput{reachable set $\mathscr{A}$ or counter example $\mathscr{C}$.}
$\RS$ := \tnm{Generator}($\Lambda,\theta,t, \epsilon$) ; \\
(res, $\mathscr{C})$ := \tnm{Verifier}($\RS$, $c$, Property) ; \\
\If{res = \texttt{accept}}
{
\Return $\RS$ ;
}
\While{ True }
{
$\epsilon$ := \tnm{Refine}($\RS,\mathscr{C}$) ;\\
$\RS$ := \tnm{Generator}($\Lambda,\theta,t, \epsilon$) ; \\
(res, $\mathscr{C})$ := \tnm{Verifier}($\RS$, $c$, Property) ; \\
\If{(res = \texttt{accept}) }
{
\Return $\RS$ ;
}
}
\end{algorithm}

\section{Reachable Sets Using Model Learning}
\label{sec:pac}
In this section, we introduce Model Learning based method to compute reachable set of a Linear Uncertain System. 
Here, we learn a probably approximately correct model of $e^{\Lambda t}$.
%
%
To efficiently learn the model and compute the reachable set, we approximate $e^{\Lambda t}$ as a linear function $w(x,\gamma,t)$ with the following template.
%
%
\begin{equation}
    \label{eq:template}
    w(x,\gamma,t) =
    e^{A_c t} x ~+
    \underbrace{
    \begin{bmatrix}
    c_{1,1} & \cdots & c_{1,p}\\
    \vdots & \ddots & \vdots\\
    c_{n,1} & \cdots & c_{n,p} \\
    \end{bmatrix}
    }_{C}
    \gamma
\end{equation}
where $x \in \theta$, $\gamma \in D_{\Lambda}$, $t \in \mathbb{R}_{\ge 0}$, $\theta$ is the initial set, $D_{\Lambda}$ is the set values of the uncertain variables and $A_c \in \Lambda$ is the \textit{mean dynamics}. 
Once such a model is learnt (\textit{i.e}, the parameters $c_{i,j}$ are learnt) --- given $\theta$ represented as a hyper-rectangle, $D_{\Lambda}$ represented as an interval set and a time step $t$, the reachable set can be computed very efficiently using interval arithmetic by plugging in the values in Equation \ref{eq:template}. 
%

\subsection{Background}
\label{sec:background}
In this subsection, we provide a brief introduction to \textit{Scenario Optimization} \cite{1632303} (as explained in \cite{xue2020pac}).
Given an optimization problem as follows [eq. (2) of \cite{xue2020pac}]:
\begin{equation}
\label{eq:def_scOpti}
\begin{aligned}
\min_{\gamma \in \Gamma \subseteq \mathbb{R}^m} \quad & c^\intercal \gamma\\
\textrm{s.t.} \quad & f_{\delta} \le 0,~~\forall \delta \in \Delta\\
\end{aligned}
\end{equation}
where $f_{\delta}(\gamma)$ are continuous continuous and convex functions over the
$m$ dimensional optimization variable $\gamma$
for every $\delta \in \Delta$.
Also, the sets $\Gamma$ and $\Delta$ are convex and closed.
The optimization problem in Equation \ref{eq:def_scOpti} can be reformulated and solved, with statistical guarantees, as shown in \cite{1632303}, as follows [eq. (3) of \cite{xue2020pac}]:
\begin{equation}
\label{eq:def_scOptiPAC}
\begin{aligned}
\min_{\gamma \in \Gamma \subseteq \mathbb{R}^m} \quad & c^\intercal \gamma\\
\textrm{s.t.} \quad & \land_{i=1}^K f_{\delta_i} \le 0,~~\forall \delta \in \Delta\\
\end{aligned}
\end{equation}

Note that the reformulation of the problem (in Equation \ref{eq:def_scOpti}), to the problem in equation \ref{eq:def_scOptiPAC} only considers finite subset out of the infinite number of constraints in Equation \ref{eq:def_scOpti}. Given, $\gamma^*$, the solution to Equation \ref{eq:def_scOptiPAC}, the statistical guarantees can be given as follows [Theorem 1 in \cite{xue2020pac}]:
\begin{theorem}
(Theorem 1 in \cite{xue2020pac})
If Equation \ref{eq:def_scOptiPAC} is feasible and attains unique optimal solution $\gamma^*$, and
\begin{equation}
    \epsilon \ge \frac{2}{K} (ln \frac{1}{\beta} + m)
\end{equation}
where $\epsilon \in (0,1)$ and $\beta \in (0,1)$ are user-chosen error and confidence levels, then with at least $1-\beta$ confidence, $\gamma^*$ satisfies all constraints in $\Delta$ but at most a fraction of probability measure $\epsilon$, \textit{i.e.}, $P \big(\{ \delta \in \Delta~|~ f_{\delta}(\gamma^*)\le 0\} \big) \le \epsilon$, where the confidence $\beta$ is the K-fold probability $p^K$ in $\Delta^k = \Delta \times \Delta \times \cdots \Delta$, which is the set to which the extracted sample $(\delta_1, \delta_2, \cdots, \delta_K)$ belongs.
\end{theorem}
\textit{The above conclusion still holds if the uniqueness of optimal solutions to Equation \ref{eq:def_scOptiPAC} is removed \cite{xue2020pac}, since a unique optimal solution can always be obtained according to Tie-break rule if multiple optimal solutions occur.}

In Subsections \ref{subsec:samplingPAC} and \ref{subsec:learnPAC}, we introduce various methods required for model learning (\textit{i.e} learning the parameters $C[i,j]$). Finally, in Subsection \ref{subsec:pacReach}, we propose our model learning based reachable set computation algorithm.

\subsection{Sampling}
\label{subsec:samplingPAC}
Given a set (represented as a \textit{generalized star}) $\theta \in \mathbb{R}^n$, the set of values of uncertain variables in $\Lambda$ represented as an interval set $D_{\Lambda} \subseteq \mathbb{R}^p$ (where $p$ is the number of uncertainties in $\Lambda$), time interval $t_{\delta} \subseteq [0,T]$; we define the following randomly chosen samples according to their corresponding probability distributions:

\begin{enumerate}
    \item Let, $\{x_1, x_2, \cdots, x_M\}$ be $M$ random samples of $\theta$, \textit{i.e.}, $x_i \in \theta$. We denote the set $\{x_1, x_2, \cdots, x_M\}$ as $\{x_i\}^M$.
    
    \item Let, $\{\gamma_1, \gamma_2, \cdots, \gamma_N\}$ be $N$ random samples of $D_{\Lambda}$, \textit{i.e.} $\gamma_j \in D_{\Lambda}$. We denote the set $\{\gamma_1, \gamma_2, \cdots, \gamma_N\}$ as $\{\gamma_j\}^N$.

    \item Let, $\{t_1, t_2, \cdots, t_O\}$ be $O$ random samples of $t_{\delta}$, \textit{i.e.} $t_k \in t_{\delta}$. We denote the set $\{t_1, t_2, \cdots, t_O\}$ as $\{t_k\}^O$.
\end{enumerate}
Given a $\gamma_j \in D_{\Lambda}$, let $A_{\gamma_j} \in \Lambda$ be the matrix obtained by assigning values $\gamma_j$ to the variables of $\Lambda$. We denote the reachable point for a given sample point $(x_i,\gamma_j,t_k)$ as $y_{i,j,k}$, \textit{i.e.}, $y_{i,j,k} = \RS(x_i,A_{\gamma_j},t_k)$. Let, $\{y_{i,j,k}\}^{M \cdot N \cdot O}=\big\{y_{i,j,k} = \RS(x_i,A_{\gamma_j},t_k)~|~ x_i \in \{x_i\}^M \land \gamma_j \in \{\gamma_j\}^N \land t_k \in \{t_k\}^O \big\}$.
Let, the shorthand of the above sampling method be denoted as $\texttt{sample} (\theta, D_{\Lambda}, t_{\delta})$ = \\ $\big(\{x_i\}^M, \{\gamma_j\}^N, \{t_k\}^O, \{y_{i,j,k}\}^{M \cdot N \cdot O} \big)$

\subsection{Learning the Model}
\label{subsec:learnPAC}
In this section, we formulate the problem of learning the parameters $C$ in Equation \ref{eq:template} as a \textit{Scenario Optimization} problem. And therefore the guarantees obtained will be statistical in nature.
Given an initial set $\theta$, a set of uncertainties $D_{\Lambda}$, and a time interval $t_\delta$, we denote its samples as $\texttt{sample} (\theta, D_{\Lambda}, t_{\delta})$ = $\big(\{x_i\}^M, \{\gamma_j\}^N, \{t_k\}^O, \{y_{i,j,k}\}^{M \cdot N \cdot O} \big)$. Using these samples, we learn the parameters $C$ (in Equation \ref{eq:template}) by formulating it as the following optimization problem:
\begin{equation}
\label{eq:scOptiPAC}
\begin{aligned}
\min_{C, \kappa} \quad & \kappa \\
\textrm{s.t.} \quad & \forall (x_i \in \{x_i\}^M, \gamma_j \in \{\gamma_j\}^N,
t_k \in \{t_k\}^O)\\
\quad & w(C,x_i,\gamma_j,t_k) - y_{i,j,k} \le \underbrace{[\kappa~ \kappa~ \cdots~ \kappa]^\intercal}_{n} \\
\quad & y_{i,j,k} - w(C,x_i,\gamma_j,t_k) \le \underbrace{[\kappa~ \kappa~ \cdots~ \kappa]^\intercal}_{n} \\
\quad & \forall_{r \in [1,n], s \in [1,p]}~ -U_C \le C[r,s] \le U_C \\
\quad & 0 \le \kappa \le U_{\kappa} \\
\end{aligned}
\end{equation}

Let, the solution to the above optimization problem be $(C^*,\kappa^*)$, where $C^*[i,j]=c^*_{i,j}$.
Let, the shorthand of this learning module be denoted as \\ \texttt{learnModel}$(\theta,D_{\Lambda},t_{\delta})=(C^*, \kappa^*)$

\subsubsection{Statistical Guarantees}
We now formalize the statistical guarantees on the solution to the \textit{Scenario Optimization} problem in Equation \ref{eq:scOptiPAC}. The statistical guarantees on the solution $(C^*, \kappa^*)$ are given in the following Theorem:
\begin{theorem}
\label{thm:pacGuarantee}
Let 
\begin{eqnarray*}
\tau(x,\gamma) &=& \{t \in t_\delta ~|~ w(C^*,x,\gamma,t) - \RS(x,A_\gamma,t) \le [\kappa^*~ \kappa^*~ \cdots~ \kappa^*]^\intercal \} \\
\chi (u) &=& \{ x \in \theta ~|~ P_t(\tau(x,u))\ge1-\epsilon_1 ~\text{with confidence}~ 1-\beta_1\} \\
\mathcal{U} &=& \{u \in D_{\Lambda}~|~ P_x(\chi(u)) \ge 1-\epsilon_2 ~\text{with confidence}~ 1-\beta_2\}
\end{eqnarray*}
then we have:
$
P_u(\{u ~|~ u \in \mathcal{U}\} \ge 1-\epsilon_3 ~\text{with confidence}~ 1-\beta_3
$
\\
where $P_t, P_x, P_u$ are the probability measures on the sets $t_\delta, \theta, U$ respectively. And 
$
\epsilon_i \ge \frac{2}{M} (ln \frac{1}{\beta_i} + np + 1)
$, for $i \in \{1, 2, 3\}$ 
\end{theorem}

\begin{proof}
Follows directly from Theorem 3 of \cite{xue2020pac}
\end{proof}

\subsection{Computing Reachable Sets from the Learnt Model}
\label{subsec:pacReach}
Once a probably approximately correct model of $e^{\Lambda t}$ is learned, we compute the reachable set by \emph{bloating} the reachable set obtained from the model by $\kappa^*$.
Therefore, the candidate reachable set would be $Bloat(e^{A_c t} \theta + C^* D_{\Lambda}, \kappa^*)$.
We refer to this reachable set as $\mathsf{modelLearn}(\theta, t)$.
%
%
%
%
%
From the optimization problem in Equation \ref{eq:scOptiPAC}, we get a statistical guarantee that the set $\theta_t$ is $\kappa^*$ close to $\RS(\theta,\Lambda,t)$. So to compute an overapproximation of $\RS(\theta,\Lambda,t)$, we bloat $\theta_t$ by $\kappa^*$ --- $\RS(\theta,\Lambda,t) \subseteq Bloat(\theta_t,\kappa^*)$, with statistical guarantee as in Theorem \ref{thm:pacGuarantee}. 
%
%
We formalize this approach in Algorithm \ref{algo:pac}, and refer to this approach as $\mathsf{modelLearn}$.

\begin{algorithm}
\caption{Reachable set computation using Model Learning}
\label{algo:pac}
\KwInput{$\Lambda$, $\theta$, $t_\delta$}
\KwOutput{$\Theta_t$}
$D_{\Lambda}$ $\leftarrow$ \text{Hyper-box representing the values of the uncertain variables in} $\Lambda$ ; \\
$(C^*,\kappa^*)$=$\texttt{learnModel}\big(Bloat(\theta,\delta_1),Bloat(D_{\Lambda},\delta_2),Bloat(t_{\delta},\delta_3)\big)$ ;\\
$\Theta_t$ = \{\} ; \\
\For{ \text{all} $t$ $\in$ $t_{\delta}$}
{
 $tmp$ = $e^{A_c t} \cdot \texttt{Box}(\theta) + C^* \cdot D_{\Lambda}$ ; \\
 $\Theta_t = \Theta_t \cup \{ Bloat(tmp,\kappa^*) \}$
}
\Return $\Theta_t$ ;
\end{algorithm}

\section{Evaluation}
\label{sec:eval}
We implemented Algorithm \ref{algo:main_algo} and \ref{algo:pac} in a \texttt{Python} based tool. 
%
%
We have used \texttt{Gurobi}\footnote{\url{https://www.gurobi.com/}} \cite{gurobi} for visualizing the reachable sets. For computing distances and to perform subset checking, we have used the tool \texttt{pypolycontain}\footnote{\url{https://github.com/sadraddini/pypolycontain}}. Additionally, we use \texttt{numpy}\footnote{\url{https://numpy.org/}} \cite{harris2020array}, \texttt{scipy}\footnote{\url{https://www.scipy.org/}} \cite{2020SciPy-NMeth} and \texttt{mpmath}\footnote{\url{http://mpmath.org/}} for several other functionalities. If the paper gets accepted, we will open source our tool. 
In the following subsections, we will evaluate our Algorithm \ref{algo:main_algo} and \ref{algo:pac} on several benchmarks, and show that our technique is indeed scalable, even for high dimensional benchmarks. Most of the benchmarks are taken from ARCH workshop\footnote{\url{https://cps-vo.org/group/ARCH/benchmarks}}. All the experiments were performed on a Lenovo ThinkPad Mobile Workstation with i7-8750H CPU with 2.20 GHz and 32GiB memory on Ubuntu 18.04 operating system (64 bit). 
We report the following: (i) performance of Algorithm \ref{algo:main_algo} using all the heuristics, $\mathsf{bloatMean}$, $\mathsf{maxSV}$, $\mathsf{encloseORH}$, $\mathsf{uniformORH}$ and $\mathsf{boxBloat}$; (ii) performance of Algorithm \ref{algo:pac}; (iii) finally, we compare the performance of Algorithm \ref{algo:main_algo} and \ref{algo:pac} mutually and with $\mathsf{Flow^*}$.
A summary of the evaluations on all the benchmarks is given in Table \ref{table:results}.

\subsection{Flight Collision}
\begin{wrapfigure}{r}{0.5\textwidth}
  \begin{center}
    \includegraphics[width=5.5cm,height=4.2cm]{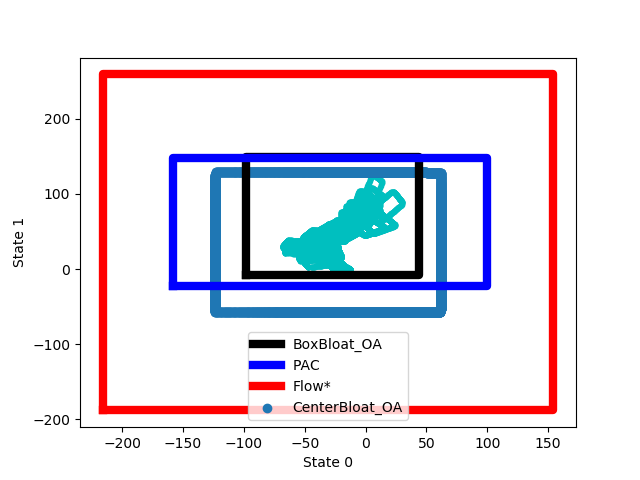}
  \end{center}
  \caption{ Over-approximate reachable sets of Flight Collison model. Dirty Blue: Computed using $\mathsf{bloatMean}$; Black: Computed using $\mathsf{boxBloat}$; Cyan: Random reachable set; Navy Blue: Computed using $\mathsf{modelLearn}$; Red: Computed using $\mathsf{Flow*}$.}
\label{fig:aircraft}
\end{wrapfigure}
\cite{10.1007/978-3-642-05089-3_35} provides an aircraft collision avoidance maneuver model. In air traffic control, collision avoidance maneuvers are used to resolve conflicting paths that arise during free flight. A complicated online trajectory prediction or maneuver planning might not be possible in a such scenarios due to limited available time. Thus an offline analysis, accounting for all possible behavior is a desired solution; making it a good candidate to evaluate our algorithm.
The model of the system is given in Section 1 of \cite{DBLP:conf/emsoft/LalP15}. The state variables in the dynamics are $x=(x_1,x_2)$, location in two dimension and $d=(d_1,d_2)$, velocity in two dimension. The differential equations are further dependent on $\omega$, angular velocity of the flight.
Given the model, and an uncertain parameter $\omega$, we assume $\omega \in [1,3]$ to evaluate our algorithm. And the initial set $\theta$ = [-1,1] $\times$ [-1,1] $\times$ [20,30] $\times$ [20,30].
We tune the \tnm{StatVer} module with $B=9000$ and $c=0.99$, yielding $K=906$ (number of samples required) from Equation \ref{eq:K}, and \emph{type I error} of $ err(B,c)=0.01$ from Equation \ref{eq:typeIerror}. And we tune \tnm{learnModel} module with $N=O=100$ (We don't consider $M$ as we are doing it for a particular time step). In Figure \ref{fig:aircraft}, we show the reachable set for time step 2000, compared with \tnm{Flow*}. The details of the time taken is given in Table \ref{table:results}.

\subsection{Five Vehicle Platoon}
In this subsection we compute reachable sets of a 15 dimensional, 5 vehicle platoon model\footnote{\url{https://ths.rwth-aachen.de/research/projects/hypro/n\_vehicle\_platoon/}}.
This benchmark is a framework of 5 autonomously driven vehicles; where one of the vehicle is a leader, located at the head of the formation, and the rest of the vehicles act as a follower.
The vehicles establish synchronization through communication via network --- exchanging information about their relative positions, relative velocities, accelerations measured with on-board sensors.
Each vehicle in the platoon is described by a 3 dimensional state vector --- relative position, its derivative and acceleration.
The main goal is to avoid collisions inside the platoon.
For a vehicle $i$: $e_i$ is its relative position, $\dot{e_i}$ is the derivative of $e_i$, and $a_i$ is the acceleration.
The state variables are: [$e_1$, $\dot{e_1}$, $a_1$, $e_2$, $\dot{e_2}$, $a_2$, $e_3$, $\dot{e_3}$, $a_3$, $e_4$, $\dot{e_4}$, $a_4$, $e_5$, $\dot{e_5}$, $a_5$]. 
The dynamics is provided in the url.
The model is fixed \emph{i.e.} not directly dependent on some parameters. To capture modelling error, we introduce perturbation of $\pm$ 2\% in the following cells: [3,4], [4,5]. And let the initial set $\theta$ = $[-1,1]^{15}$.
We tune the \tnm{StatVer} module with $B=9000$ and $c=0.99$, yielding $K=906$ (number of samples required) from Equation \ref{eq:K}, and \emph{type I error} of $ err(B,c)=0.01$ from Equation \ref{eq:typeIerror}. And we tune \tnm{learnModel} module with $N=O=300$ (We don't consider $M$ as we are doing it for a particular time step). In Figure \ref{fig:five_space} (Left), we show the reachable set for time step 2000, compared with \tnm{Flow*}. The details of the time taken is given in Table \ref{table:results}.

\subsection{Spacecraft Rendezvous}
\cite{ARCH17:Verifying_safety_of_an} provides a linear model of autonomous maneuver of a spacecraft navigating to and approaching other spacecraft.
%
%
The dynamics of the two spacecraft in orbit—the target and the chaser—are derived from Kepler’s laws. The two spacecrafts are assumed to be in same the orbital plane. The target is assumed to move on a circular orbit. The linearized model is given in Section 3.2 of \cite{ARCH18:Lane_change_maneuver_for}. The equations depend on parameters like $n=\sqrt{\frac{\mu}{r}}$, where $\mu=3.986 \times 10^{14}m^3/s^2$ and $r=42164km$ and $m_c=500kg$, the mass of the spacecraft.
For our experiments, we assume $m_c \in [475,525]$, as a result of perturbation due to measurement and $n \in [97200, 97250]$. And let the initial $\theta$ = [-1,1] $\times$ [-1,1] $\times$ [0,0] $\times$ [0,0] $\times$ [1,1] $\times$ [1,1]. 
We tune the \tnm{StatVer} module with $B=9000$ and $c=0.99$, yielding $K=906$ (number of samples required) from Equation \ref{eq:K}, and \emph{type I error} of $ err(B,c)=0.01$ from Equation \ref{eq:typeIerror}. And we tune \tnm{learnModel} module with $N=O=200$ (We don't consider $M$ as we are doing it for a particular time step). In Figure \ref{fig:five_space} (Right), we show the reachable set for time step 1901, compared with \tnm{Flow*}. \textit{Note that Flow* stopped after 1901 time steps, as it could not compute the flowpipes anymore. Our tool showed no such limitations}. The details of the time taken is given in Table \ref{table:results}.

\begin{figure}[ht]
\centering
\includegraphics[width=12.5cm,height=4cm]{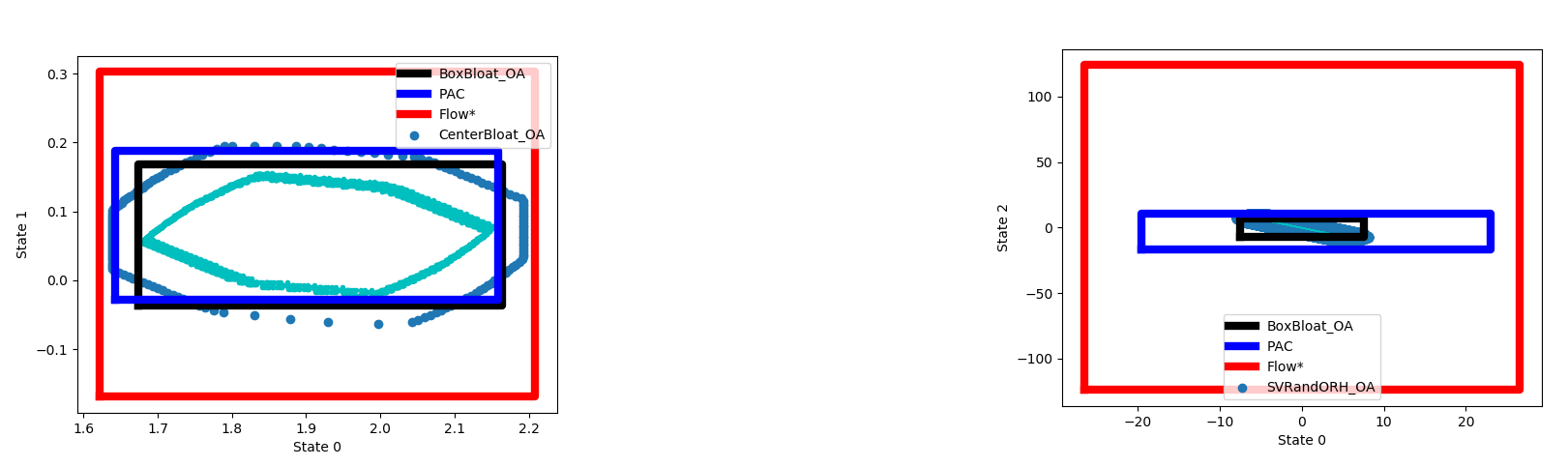}
\caption{ \textbf{Left}: Over-approximate reachable sets of 5 Vehicle Platoon model. Dirty Blue: Computed using $\mathsf{bloatMean}$; Black: Computed using $\mathsf{boxBloat}$; Cyan: Random reachable set; Navy Blue: Computed using $\mathsf{modelLearn}$; Red: Computed using $\mathsf{Flow*}$. \textbf{Right}: Over-approximate reachable sets of Spacecraft Rendezvous model. Dirty Blue: Computed using $\mathsf{maxSV}$; Black: Computed using $\mathsf{boxBloat}$; Cyan: Random reachable set; Navy Blue: Computed using $\mathsf{modelLearn}$; Red: Computed using $\mathsf{Flow*}$.}
\label{fig:five_space}
\end{figure}

\subsection{More Benchmarks}
We evaluated our algorithms on several other benchmarks ---
(i) An \emph{Anaesthesia Delivery Model} \cite{ARCH15:Benchmark_Problem_PK_PD_Model}, a 5 dimensional model;
(ii) a 16 dimensional \emph{Autonomous Quadrotors Model} \cite{Kaynama2014BenchmarkFE};
(iii) a 5 dimensional \emph{Motor Transmission Drive System} \cite{ARCH15:Motor_Transmission_Drive_System_Benchmark};
(iv) A 10 dimensional model of a \textit{Self-balancing Two-wheeled Robot} \cite{ARCH15:Benchmark_Reachability_on_model};
(v) a 4 dimensional model of an \textit{Adaptive Cruise Control (ACC)} \cite{7349170}.
Due to lack of space, the plots for these benchmarks are not provided. The summary of all the evaluation is provided in Table \ref{table:results}.

\begin{table}
\caption{Evaluation on Benchmarks. \textbf{Dim}: Number of dimensions of the benchmark; \textbf{Hypothesis Testing}: Time taken (in seconds) by various hypothesis testing based methods (The sub-columns are as follows. \textbf{Mean}: $\mathsf{bloatMean}$, \textbf{SV}: $\mathsf{maxSV}$, \textbf{ORH}: $\mathsf{encloseORH}$, \textbf{Uniform}: $\mathsf{uniformORH}$,
\textbf{Box}: $\mathsf{boxBloat}$.); \textbf{Learn}: Time taken (in seconds) by $\mathsf{modelLearn}$ method; \textbf{Flow*}: Time taken (in seconds) by \texttt{Flow*}.}\label{tab1}

\begin{tabular}{|l|c|c|c|c|c|c|c|c|}
\hline
\cline{3-7}
\multicolumn{1}{|c|}{}& \multicolumn{1}{c|}{} & \multicolumn{5}{c|}{\textbf{Hypothesis Testing}} &
\multicolumn{1}{c|}{} &
\multicolumn{1}{c|}{}
\\
\hline
\textbf{Benchmark} & \textbf{Dim} & {\textbf{Mean}} & {\textbf{SV}} & {\textbf{ORH}} & {\textbf{Uniform}} & {\textbf{Box}} & \textbf{Learn} & \textbf{Flow*}\\
\hline
Adaptive Cruise Control & 4 & 9.28 & 14.15 & 11.16 & 11.58 & \textbf{1.57} & 3.97 & 41.35\\
Flight Collision & 4 & 12.35 & 7.28 & 10.36 & 7.59 & \textbf{2.045} & 4.077 & 922.48\\
Anaesthesia & 5 & 7.39 & 12.44 & 31.53 & 17.89 & \textbf{1.78} & 8.13 & 24.73\\
Motor-Transmission & 5 & 17.95 & 12.73 & 13.02 & 12.78 & \textbf{1.72} & 9.18 & 19.51\\
Spacecraft Rendezvous & 6 & 25.18 & 19.1 & 19.477 & 18.98 & \textbf{3.05} & 24.57 & 2176.54\\
2-wheeled Robot & 10 &  112.68 & - & 119.52 & 109.88 & \textbf{7.604} & 90.51 & 50.48\\
5 Vehicle Platoon & 15 & 250.1 & - & 231.41 & 214.38 & \textbf{11.455} & 147.9 & 1745.47\\
Quadrotor & 16 &  330.1 & - & 392.37 & 350.84 & \textbf{5.74} & 152.73 & 39.62\\
\hline
\end{tabular}
\label{table:results}
\centering
\end{table}


As observed in Table~\ref{table:results}, the time taken for computing candidate reachable sets with statistical guarantees is much less than compared to the nonlinear reachability tools like Flow*.
Notice that the time taken by $\mathsf{boxBloat}$ is always the least followed by $\mathsf{modelLearn}$ in almost all of these examples.
Also, while $\mathsf{encloseORH}$ would give conservative overapproximation in some instances, the computation required to compute it is two orders of magnitude compared to $\mathsf{boxBloat}$ or $\mathsf{modelLearn}$.

\section{Conclusions}
In this paper, we presented two techniques to compute reachable set of uncertain linear systems that provide probabilistic guarantees.
The first technique involves performing statistical hypothesis testing whereas the second involves model learning.
One of the heuristic proposed in computing the candidate for statistical hypothesis testing provides us with a conservative overarpproximation of the reachable set when the uncertain system satisfies certain conditions.
Evaluation of our techniques on various benchmarks shows that our approach can compute reachable set artifacts for high dimensional systems relatively quickly.
While verification of CPS has received significant attention, the work on quantifying the effects of uncertainty on the safety specification has not been studied extensively.

In this paper, we attempted to bridge the gap between computational efficiency and statistical guarantees for computing the reachable set artifacts.
This would help us in comparing the robustness of the safety specification with respect to model perturbations.
In future, we intend to apply this technique for synthesizing controllers that satisfy the safety specification even under model perturbations.

\bibliographystyle{splncs04}
\bibliography{main}

\begin{thebibliography}{10}
\providecommand{\url}[1]{\texttt{#1}}
\providecommand{\urlprefix}{URL }
\providecommand{\doi}[1]{https://doi.org/#1}

\bibitem{ARCH15:An_Introduction_to_CORA}
Althoff, M.: An introduction to cora 2015. In: Frehse, G., Althoff, M. (eds.)
  ARCH14-15. 1st and 2nd International Workshop on Applied veRification for
  Continuous and Hybrid Systems. EPiC Series in Computing, vol.~34, pp.
  120--151. EasyChair (2015). \doi{10.29007/zbkv},
  \url{https://easychair.org/publications/paper/xMm}

\bibitem{10.1145/1967701.1967717}
Althoff, M., Le~Guernic, C., Krogh, B.H.: Reachable set computation for
  uncertain time-varying linear systems. In: Proceedings of the 14th
  International Conference on Hybrid Systems: Computation and Control. p.
  93–102. HSCC '11, Association for Computing Machinery, New York, NY, USA
  (2011). \doi{10.1145/1967701.1967717},
  \url{https://doi.org/10.1145/1967701.1967717}

\bibitem{ashok2019pac}
Ashok, P., Křetínský, J., Weininger, M.: Pac statistical model checking for
  markov decision processes and stochastic games (2019)

\bibitem{bak2017hylaa}
Bak, S., Duggirala, P.S.: Hylaa: A tool for computing simulation-equivalent
  reachability for linear systems. In: Proceedings of the 20th International
  Conference on Hybrid Systems: Computation and Control. pp. 173--178 (2017)

\bibitem{1632303}
{Calafiore}, G.C., {Campi}, M.C.: The scenario approach to robust control
  design. IEEE Transactions on Automatic Control  \textbf{51}(5),  742--753
  (2006). \doi{10.1109/TAC.2006.875041}

\bibitem{ARCH17:Verifying_safety_of_an}
Chan, N., Mitra, S.: Verifying safety of an autonomous spacecraft rendezvous
  mission. In: Frehse, G., Althoff, M. (eds.) ARCH17. 4th International
  Workshop on Applied Verification of Continuous and Hybrid Systems. EPiC
  Series in Computing, vol.~48, pp. 20--32. EasyChair (2017).
  \doi{10.29007/thb4}, \url{https://easychair.org/publications/paper/S2V}

\bibitem{ARCH15:Motor_Transmission_Drive_System_Benchmark}
Chen, H., Mitra, S., Tian, G.: Motor-transmission drive system: a benchmark
  example for safety verification. In: Frehse, G., Althoff, M. (eds.)
  ARCH14-15. 1st and 2nd International Workshop on Applied veRification for
  Continuous and Hybrid Systems. EPiC Series in Computing, vol.~34, pp. 9--18.
  EasyChair (2015). \doi{10.29007/ct87},
  \url{https://easychair.org/publications/paper/cwl}

\bibitem{7809839}
{Chen}, X., {Sankaranarayanan}, S.: Decomposed reachability analysis for
  nonlinear systems. In: 2016 IEEE Real-Time Systems Symposium (RTSS). pp.
  13--24 (2016). \doi{10.1109/RTSS.2016.011}

\bibitem{10.1007/978-3-642-39799-8_18}
Chen, X., {\'A}brah{\'a}m, E., Sankaranarayanan, S.: Flow*: An analyzer for
  non-linear hybrid systems. In: Sharygina, N., Veith, H. (eds.) Computer Aided
  Verification. pp. 258--263. Springer Berlin Heidelberg, Berlin, Heidelberg
  (2013)

\bibitem{chen2015pac}
Chen, Y.F., Hsieh, C., Lengál, O., Lii, T.J., Tsai, M.H., Wang, B.Y., Wang,
  F.: Pac learning-based verification and model synthesis (2015)

\bibitem{10.1007/978-3-642-24372-1_1}
Clarke, E.M., Zuliani, P.: Statistical model checking for cyber-physical
  systems. In: Bultan, T., Hsiung, P.A. (eds.) Automated Technology for
  Verification and Analysis. pp. 1--12. Springer Berlin Heidelberg, Berlin,
  Heidelberg (2011)

\bibitem{7945001}
{Diwakaran}, R.D., {Sankaranarayanan}, S., {Trivedi}, A.: Analyzing
  neighborhoods of falsifying traces in cyber-physical systems. In: 2017
  ACM/IEEE 8th International Conference on Cyber-Physical Systems (ICCPS). pp.
  109--120 (2017)

\bibitem{6658604}
{Duggirala}, P.S., {Mitra}, S., {Viswanathan}, M.: Verification of annotated
  models from executions. In: 2013 Proceedings of the International Conference
  on Embedded Software (EMSOFT). pp. 1--10 (2013).
  \doi{10.1109/EMSOFT.2013.6658604}

\bibitem{10.1007/978-3-662-46681-0_5}
Duggirala, P.S., Mitra, S., Viswanathan, M., Potok, M.: C2e2: A verification
  tool for stateflow models. In: Baier, C., Tinelli, C. (eds.) Tools and
  Algorithms for the Construction and Analysis of Systems. pp. 68--82. Springer
  Berlin Heidelberg, Berlin, Heidelberg (2015)

\bibitem{10.1007/978-3-319-41528-4_26}
Duggirala, P.S., Viswanathan, M.: Parsimonious, simulation based verification
  of linear systems. In: Chaudhuri, S., Farzan, A. (eds.) Computer Aided
  Verification. pp. 477--494. Springer International Publishing, Cham (2016)

\bibitem{10.1007/978-3-642-24690-6_13}
Eggers, A., Ramdani, N., Nedialkov, N., Fr{\"a}nzle, M.: Improving sat modulo
  ode for hybrid systems analysis by combining different enclosure methods. In:
  Barthe, G., Pardo, A., Schneider, G. (eds.) Software Engineering and Formal
  Methods. pp. 172--187. Springer Berlin Heidelberg, Berlin, Heidelberg (2011)

\bibitem{fan2017dryvrdatadriven}
Fan, C., Qi, B., Mitra, S., Viswanathan, M.: Dryvr:data-driven verification and
  compositional reasoning for automotive systems (2017)

\bibitem{farhadsefat2011norms}
Farhadsefat, R., Rohn, J., Lotfi, T.: Norms of interval matrices  (2011)

\bibitem{frehse2005phaver}
Frehse, G.: Phaver: Algorithmic verification of hybrid systems past hytech. In:
  International workshop on hybrid systems: computation and control. pp.
  258--273. Springer (2005)

\bibitem{frehse2011spaceex}
Frehse, G., Le~Guernic, C., Donz{\'e}, A., Cotton, S., Ray, R., Lebeltel, O.,
  Ripado, R., Girard, A., Dang, T., Maler, O.: Spaceex: Scalable verification
  of hybrid systems. In: International Conference on Computer Aided
  Verification. pp. 379--395. Springer (2011)

\bibitem{ARCH15:Benchmark_Problem_PK_PD_Model}
Gan, V., Dumont, G., Mitchell, I.: Benchmark problem: A pk/pd model and safety
  constraints for anesthesia delivery. In: Frehse, G., Althoff, M. (eds.)
  ARCH14-15. 1st and 2nd International Workshop on Applied veRification for
  Continuous and Hybrid Systems. EPiC Series in Computing, vol.~34, pp.~1--8.
  EasyChair (2015). \doi{10.29007/8drm},
  \url{https://easychair.org/publications/paper/R8kX}

\bibitem{ghosh2019robust}
Ghosh, B., Duggirala, P.S.: Robust reachable set: Accounting for uncertainties
  in linear dynamical systems. ACM Transactions on Embedded Computing Systems
  (TECS)  \textbf{18}(5s),  1--22 (2019)

\bibitem{girard2005reachability}
Girard, A.: Reachability of uncertain linear systems using zonotopes. In:
  International Workshop on Hybrid Systems: Computation and Control. pp.
  291--305. Springer (2005)

\bibitem{gurobi}
Gurobi~Optimization, L.: Gurobi optimizer reference manual (2020),
  \url{http://www.gurobi.com}

\bibitem{harris2020array}
Harris, C.R., Millman, K.J., van~der Walt, S.J., Gommers, R., Virtanen, P.,
  Cournapeau, D., Wieser, E., Taylor, J., Berg, S., Smith, N.J., Kern, R.,
  Picus, M., Hoyer, S., van Kerkwijk, M.H., Brett, M., Haldane, A., del
  R{'{\i}}o, J.F., Wiebe, M., Peterson, P., G{'{e}}rard-Marchant, P., Sheppard,
  K., Reddy, T., Weckesser, W., Abbasi, H., Gohlke, C., Oliphant, T.E.: Array
  programming with {NumPy}. Nature  \textbf{585}(7825),  357--362 (Sep 2020).
  \doi{10.1038/s41586-020-2649-2},
  \url{https://doi.org/10.1038/s41586-020-2649-2}

\bibitem{ARCH15:Benchmark_Reachability_on_model}
Heinz, T., Oehlerking, J., Woehrle, M.: Benchmark: Reachability on a model with
  holes. In: Frehse, G., Althoff, M. (eds.) ARCH14-15. 1st and 2nd
  International Workshop on Applied veRification for Continuous and Hybrid
  Systems. EPiC Series in Computing, vol.~34, pp. 31--36. EasyChair (2015).
  \doi{10.29007/cv59}, \url{https://easychair.org/publications/paper/sPgl}

\bibitem{10.1007/978-3-642-03845-7_15}
Jha, S.K., Clarke, E.M., Langmead, C.J., Legay, A., Platzer, A., Zuliani, P.: A
  bayesian approach to model checking biological systems. In: Degano, P.,
  Gorrieri, R. (eds.) Computational Methods in Systems Biology. pp. 218--234.
  Springer Berlin Heidelberg, Berlin, Heidelberg (2009)

\bibitem{Kaynama2014BenchmarkFE}
Kaynama, S., Tomlin, C.: Benchmark: Flight envelope protection in autonomous
  quadrotors (2014)

\bibitem{ARCH18:Lane_change_maneuver_for}
Kekatos, N., He\{\textbackslash{}ss\}, D., Frehse, G.: Lane change maneuver for
  autonomous vehicles (benchmark proposal). In: Frehse, G. (ed.) ARCH18. 5th
  International Workshop on Applied Verification of Continuous and Hybrid
  Systems. EPiC Series in Computing, vol.~54, pp. 229--241. EasyChair (2018).
  \doi{10.29007/5hxt}, \url{https://easychair.org/publications/paper/Hx1f}

\bibitem{10.1007/978-3-662-46681-0_15}
Kong, S., Gao, S., Chen, W., Clarke, E.: dreach: $\delta$-reachability analysis
  for hybrid systems. In: Baier, C., Tinelli, C. (eds.) Tools and Algorithms
  for the Construction and Analysis of Systems. pp. 200--205. Springer Berlin
  Heidelberg, Berlin, Heidelberg (2015)

\bibitem{DBLP:conf/emsoft/LalP15}
Lal, R., Prabhakar, P.: Bounded error flowpipe computation of parameterized
  linear systems. In: Girault, A., Guan, N. (eds.) 2015 International
  Conference on Embedded Software, {EMSOFT} 2015, Amsterdam, Netherlands,
  October 4-9, 2015. pp. 237--246. {IEEE} (2015).
  \doi{10.1109/EMSOFT.2015.7318279},
  \url{https://doi.org/10.1109/EMSOFT.2015.7318279}

\bibitem{10.1007/978-3-642-16612-9_11}
Legay, A., Delahaye, B., Bensalem, S.: Statistical model checking: An overview.
  In: Barringer, H., Falcone, Y., Finkbeiner, B., Havelund, K., Lee, I., Pace,
  G., Ro{\c{s}}u, G., Sokolsky, O., Tillmann, N. (eds.) Runtime Verification.
  pp. 122--135. Springer Berlin Heidelberg, Berlin, Heidelberg (2010)

\bibitem{7349170}
{Nilsson}, P., {Hussien}, O., {Balkan}, A., {Chen}, Y., {Ames}, A.D.,
  {Grizzle}, J.W., {Ozay}, N., {Peng}, H., {Tabuada}, P.:
  Correct-by-construction adaptive cruise control: Two approaches. IEEE
  Transactions on Control Systems Technology  \textbf{24}(4),  1294--1307
  (2016)

\bibitem{park2020pac}
Park, S., Bastani, O., Matni, N., Lee, I.: Pac confidence sets for deep neural
  networks via calibrated prediction (2020)

\bibitem{10.1007/978-3-642-05089-3_35}
Platzer, A., Clarke, E.M.: Formal verification of curved flight collision
  avoidance maneuvers: A case study. In: Cavalcanti, A., Dams, D.R. (eds.) FM
  2009: Formal Methods. pp. 547--562. Springer Berlin Heidelberg, Berlin,
  Heidelberg (2009)

\bibitem{10.1145/3049797.3049804}
Roohi, N., Wang, Y., West, M., Dullerud, G.E., Viswanathan, M.: Statistical
  verification of the toyota powertrain control verification benchmark. In:
  Proceedings of the 20th International Conference on Hybrid Systems:
  Computation and Control. p. 65–70. HSCC '17, Association for Computing
  Machinery, New York, NY, USA (2017). \doi{10.1145/3049797.3049804},
  \url{https://doi.org/10.1145/3049797.3049804}

\bibitem{6987596}
{Rwth}, X.C., {Sankaranarayanan}, S., {Ábrahám}, E.: Under-approximate
  flowpipes for non-linear continuous systems. In: 2014 Formal Methods in
  Computer-Aided Design (FMCAD). pp. 59--66 (2014).
  \doi{10.1109/FMCAD.2014.6987596}

\bibitem{sadraddini2019linear}
Sadraddini, S., Tedrake, R.: Linear encodings for polytope containment problems
  (2019)

\bibitem{10.1007/11513988_26}
Sen, K., Viswanathan, M., Agha, G.: On statistical model checking of stochastic
  systems. In: Etessami, K., Rajamani, S.K. (eds.) Computer Aided Verification.
  pp. 266--280. Springer Berlin Heidelberg, Berlin, Heidelberg (2005)

\bibitem{10.1007/3-540-36580-X_35}
Stursberg, O., Krogh, B.H.: Efficient representation and computation of
  reachable sets for hybrid systems. In: Maler, O., Pnueli, A. (eds.) Hybrid
  Systems: Computation and Control. pp. 482--497. Springer Berlin Heidelberg,
  Berlin, Heidelberg (2003)

\bibitem{10.1007/978-3-319-02444-8_37}
Testylier, R., Dang, T.: Nltoolbox: A library for reachability computation of
  nonlinear dynamical systems. In: Van~Hung, D., Ogawa, M. (eds.) Automated
  Technology for Verification and Analysis. pp. 469--473. Springer
  International Publishing, Cham (2013)

\bibitem{2020SciPy-NMeth}
Virtanen, P., Gommers, R., Oliphant, T.E., Haberland, M., Reddy, T.,
  Cournapeau, D., Burovski, E., Peterson, P., Weckesser, W., Bright, J., {van
  der Walt}, S.J., Brett, M., Wilson, J., Millman, K.J., Mayorov, N., Nelson,
  A.R.J., Jones, E., Kern, R., Larson, E., Carey, C.J., Polat, {\.I}., Feng,
  Y., Moore, E.W., {VanderPlas}, J., Laxalde, D., Perktold, J., Cimrman, R.,
  Henriksen, I., Quintero, E.A., Harris, C.R., Archibald, A.M., Ribeiro, A.H.,
  Pedregosa, F., {van Mulbregt}, P., {SciPy 1.0 Contributors}: {{SciPy} 1.0:
  Fundamental Algorithms for Scientific Computing in Python}. Nature Methods
  \textbf{17},  261--272 (2020). \doi{10.1038/s41592-019-0686-2}

\bibitem{wang2019statistical}
Wang, Y., Zarei, M., Bonakdarpour, B., Pajic, M.: Statistical verification of
  hyperproperties for cyber-physical system (2019)

\bibitem{8882768}
{Xue}, B., {Liu}, Y., {Ma}, L., {Zhang}, X., {Sun}, M., {Xie}, X.: Safe inputs
  approximation for black-box systems. In: 2019 24th International Conference
  on Engineering of Complex Computer Systems (ICECCS). pp. 180--189 (2019).
  \doi{10.1109/ICECCS.2019.00027}

\bibitem{xue2020pac}
Xue, B., Zhang, M., Easwaran, A., Li, Q.: Pac model checking of black-box
  continuous-time dynamical systems (2020)

\bibitem{10.5555/647771.760735}
Younes, H.L.S., Simmons, R.G.: Probabilistic verification of discrete event
  systems using acceptance sampling. In: Proceedings of the 14th International
  Conference on Computer Aided Verification. p. 223–235. CAV ’02,
  Springer-Verlag, Berlin, Heidelberg (2002)

\bibitem{10.1145/3365365.3382209}
Zarei, M., Wang, Y., Pajic, M.: Statistical verification of learning-based
  cyber-physical systems. In: Proceedings of the 23rd International Conference
  on Hybrid Systems: Computation and Control. HSCC '20, Association for
  Computing Machinery, New York, NY, USA (2020). \doi{10.1145/3365365.3382209},
  \url{https://doi.org/10.1145/3365365.3382209}

\end{thebibliography}

\appendix

\section{Generator: Computing Candidate Reachable Set}
\label{appx:generator}

\subsection{Bloating the Reachable Set of Mean Dynamics}
\label{appxsubsec:centerBased}
Following is an example illustrating the technique in Section \ref{subsec:centerBased}:

\begin{example}
\label{ex:reachmean}
Consider the uncertain linear system $\dot{x} = \Lambda x$ where 
$\Lambda = \\ \begin{bmatrix}
    1  & [-2,-1.8] \\
    0   & -2 
\end{bmatrix}$. For this system, the \emph{mean} dynamics would be 
$\dot{x} = \begin{bmatrix}
    1  & -1.9 \\
    0   & -2 
\end{bmatrix} x$.
To compute the candidate reachable set at time $t$, we first compute the reachable set of the mean dynamics at time $t$, denoted as $\RS_{mean}$. 
We then generate several samples dynamics from $\Lambda$ according to the distribution $\mu$ (say $N$ of them denoted as $A_1, A_2, \ldots, A_N$ ) and compute the reachable set for each of them ($\RS_1, \RS_2, \ldots, \RS_N$ respectively).
We compute each of these reachable sets using the generalized star representation as described in Definition~\ref{def:reachSetStar}
The Hausdorff distance between the reachable sets is computed using linear encoding described in~\cite{sadraddini2019linear}.
We then bloat the mean reachable set $\RS_{mean}$ by $\epsilon$ added to the maximum of the Hausdorff distance.
\end{example}

\subsection{Sampling Based on Singular Values}
\label{appxsubsec:maxSVBased}

Following is an example illustrating the technique in Section \ref{subsec:maxSVBased}:

\begin{example}
\label{ex:reachSV}
Building on top of Example~\ref{ex:reachmean}, we compute the matrix with the maximum singular value in the interval matrix $\begin{bmatrix}
    1  & [-2,-1.8] \\
    0   & -2 
\end{bmatrix}$
and perform the rest of the steps similar to that is Example~\ref{ex:reachmean}.
\end{example}

\subsection{Convex Overapproximation Using Oriented Rectangular Hulls}
\label{subsec:orhBas
ed}

Figure \ref{fig:exORH} shows an \emph{Oriented Rectangular Hull} (green) computed around a set of 5 reachable sets (cyan) for one of the benchmarks considered in this paper.

\begin{figure}[ht]
\centering
\includegraphics[width=8.2cm,height=6cm]{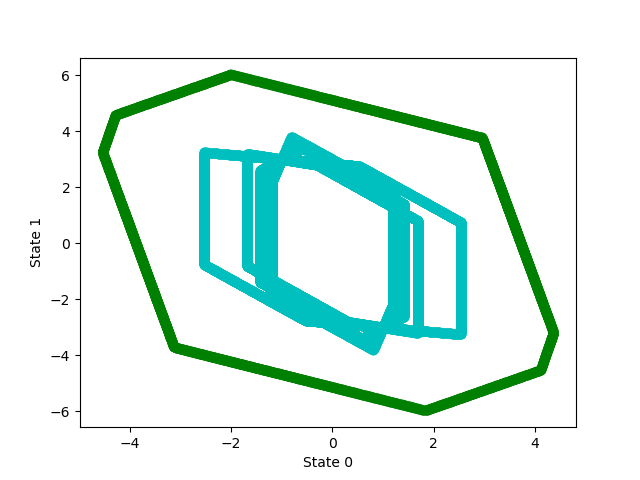}
\caption{Computed Oriented Rectangular Hull in green, around the reachable sets in blue}
\label{fig:exORH}
\end{figure}

\begin{example}
\label{ex:reachORH}

Building on Example~\ref{ex:reachmean}, for computing the oriented rectangular hull candidate reachable set, we sample $N$ dynamics from 
$\begin{bmatrix}
    1  & [-2,-1.8] \\
    0   & -2 
\end{bmatrix}$
at random, compute the template directions from oriented rectangular hull, and compute the template based overapproximation of each of the sample reachable sets computed.
\end{example}

\subsection{Structure Guided Reachable Set Computation}
\label{appxsubsec:sepBased}
\subsubsection{Conservative Overapproximation for Special Conditions}
\label{appxsubsubsec:conservative}
In this section, we prove that if all the variables in $\Lambda$ are LME preserving, then $\texttt{uniformORH}(\Lambda, \theta, t)$ is an overapproximation of the reachable set of all the dynamics in $\Lambda$ for initial set $\theta$.
This proof has two parts.
First, we show that if all variables are LME preserving, then the matrix exponential is also an LME.
Second, if the matrix exponential is an LME, then the union of matrix exponential of all samples in $\Lambda$ can be obtained by computing convex hull of the matrix exponential of vertices of $\Lambda$.
As a consequence, the convex hull of the reachable set of the vertices of $\Lambda$ contains the union of reachable set of all samples in $\Lambda$.

\begin{definition}[From~\cite{ghosh2019robust}]
\label{def:LMEclosed}
Given a linear matrix expression (LME) $M_{\Lambda} = M_0 + M_1 y_1 + \ldots + M_k y_k$, if
\begin{enumerate}
\item $supp(M_0) \times supp(M_0) \leq supp(M_0)$.
\item $\forall 1 \leq j \leq k, supp(M_0) \times supp(M_j) \leq supp(M_j)$ and \\ $ supp(M_j) \times supp(M_0) \leq supp(M_j)$.
\item $\forall 1 \leq j,l \leq k, supp(M_j) \times supp(M_l) = \textbf{0}$.
\end{enumerate}
then $M_{\Lambda}^n$ is also an LME for all $n \geq 1$. We call such LMEs as closed under multiplication.
\end{definition}

\begin{lemma}
\label{lem:matrixexpo}
If the uncertain linear system satisfies the conditions in Definition~\ref{def:LMEclosed}, then $e^{M_{\Lambda}}$ is also an LME.
Here, the matrix exponential $e^{M} = I + \frac{M}{1!} + \frac{M^2}{2!} + \frac{M^3}{3!} + \ldots$.
\end{lemma}
\begin{proof}
Follows from the fact that  $M_{\Lambda}^{n}$ is also an LME for all $n \geq 1$ and LMEs are closed under addition operation.
\end{proof}

\begin{lemma}
\label{lem:closure}
Given an uncertain linear system $\Lambda$ with closed LME, and $\overline{\gamma_1}, \overline{\gamma_2}$ be two valuations of variables in $D_{\Lambda}$, we have
$e^{\Lambda_{\beta \overline{\gamma_1} + (1-\beta) \overline{\gamma_2}}} = \beta e^{\Lambda_{\overline{\gamma_1}}} + (1-\beta) e^{\Lambda_{\overline{\gamma_2}}}.$
\end{lemma}
\begin{proof}
From Lemma~\ref{lem:matrixexpo}, we know that $e^{\Lambda}$ is an LME, say $Q_0 + Q_1 y_1 + \ldots Q_k y_k$.
Evaluating this LME over valuation $\overline{\gamma_1}$ would yield 
$$
e^{\Lambda_{\overline{\gamma_1}}} = Q_0 + Q_1 \overline{\gamma_1}[1] + \ldots Q_k \overline{\gamma_1}[k].
$$
and evaluating the LME over $\overline{\gamma_2}$ yields
$$
e^{\Lambda_{\overline{\gamma_2}}} = Q_0 + Q_1 \overline{\gamma_2}[1] + \ldots Q_k \overline{\gamma_2}[k].
$$
Observing these expressions, it is trivial to observe that the lemma holds.
\end{proof}

\begin{theorem}
\label{thm:first}
Given an uncertain linear system $\Lambda$ that is an LME closed under multiplication, and let $A_1, A_2, \ldots, A_M$ be the sample dynamics that are a result of assigning variables to the vertices of $D_{\Lambda}$.

Given an initial set $\theta$, let $\RS_i$ denote $\RS_{\theta}(A_i, t)$ for $1 \leq i \leq M$. Then we have 
$$
\mathsf{ConvexHull}(\RS_1, \ldots, \RS_M) \supseteq \bigcup_{A \in \Lambda} \RS_{\theta}(A, t).
$$
\end{theorem}
\begin{proof}
consider an state $x_0 \in \theta$ and $A \in \Lambda$. The state reached by $x_0$ after time $t$ with the dynamics $\dot{x} = Ax$ is $e^{At}x_0$.

Now, since $A \in \Lambda$ and $A_1, \ldots, A_M$ are the vertices of $\Lambda$, from Lemma~\ref{lem:closure}, we have $\exists \lambda_1, \ldots, \lambda_M$ such that 
$0 \leq \lambda_i \leq 1$ and $\Sigma_{i=1}^{M} \lambda_i = 1$ such that 
$$
e^{At} = \lambda_1 e^{A_1 t} + \ldots + \lambda_M e^{A_M t}.
$$
Therefore, 
$$
e^{At}x_0 = \lambda_1 e^{A_1 t}x_0 + \ldots + \lambda_M e^{A_M t}x_0.
$$
Since $\lambda_1 e^{A_1 t}x_0 + \ldots + \lambda_M e^{A_M t}x_0 \in \mathsf{ConvexHull}(\RS_1, \ldots, \RS_M)$, the proof is complete.
\end{proof}

\begin{example}
Building on Example~\ref{ex:reachmean}, consider the uncertain linear system  \\
$\begin{bmatrix}
    1  & [-2,-1.8] \\
    0   & -2 
\end{bmatrix}$.
Notice that this uncertain system satisfies the conditions for closure of LME under multiplication. 
Therefore, we compute the reachable set of two dynamics 
$\begin{bmatrix}
    1  & -2 \\
    0   & -2 
\end{bmatrix}$ and 
$\begin{bmatrix}
    1  & -1.8 \\
    0   & -2 
\end{bmatrix}$
and the convex hull of these two reachable sets contains the reachable set of the uncertain system.
\end{example}

\end{document}